\documentclass[11pt]{article}
\usepackage{style}
\usepackage{subcaption}
\usepackage{qcircuit}
\pdfoutput=1

\AtEveryBibitem{ %
    \clearfield{day}
    \clearfield{month}
    \clearfield{series}
    \clearfield{venue}
    \clearname{editor}
    \clearlist{publisher}
    \clearlist{location} %
    \clearfield{venue}
    \clearfield{issn}
    \clearfield{isbn}
    \clearfield{urldate}
    \clearfield{eventdate}
    \clearfield{pages}
    \clearfield{number}
    \clearfield{volume}
}

\newcommand{\PTIME}{\mathsf{P}}
\newcommand{\NEXP}{\mathsf{NEXP}}
\newcommand{\PSPACE}{\mathsf{PSPACE}}
\newcommand{\PP}{\mathsf{PP}}
\newcommand{\PPpoly}{\mathsf{PP/poly}}
\newcommand{\IP}{\mathsf{IP}}
\newcommand{\CH}{\mathsf{CH}}
\newcommand{\QIP}{\mathsf{QIP}}
\newcommand{\uQIP}{\mathsf{QIP}_{\mathsf{unent}}}
\newcommand{\QMApoly}{\mathsf{QMA/poly}}
\newcommand{\QAMpoly}{\mathsf{QAM/poly}}
\newcommand{\uQMACMpoly}{\mathsf{QMACM_{unent}/poly}}
\newcommand{\QMAM}{\mathsf{QMAM}}
\newcommand{\QMACM}{\mathsf{QMACM}}
\newcommand{\BPP}{\mathsf{BPP}}
\newcommand{\BQP}{\mathsf{BQP}}
\newcommand{\BQPqpoly}{\mathsf{BQP/qpoly}}
\newcommand{\QCIP}{\mathsf{QCIP}}
\newcommand{\QCMA}{\mathsf{QCMA}}
\newcommand{\QCAM}{\mathsf{QCAM}}
\newcommand{\AM}{\mathsf{AM}}
\newcommand{\QAM}{\mathsf{QAM}}
\newcommand{\QMA}{\mathsf{QMA}}
\newcommand{\uQMACM}{\mathsf{QMACM}_{\mathsf{unent}}}
\newcommand{\Ayes}{A_{\rm yes}}
\newcommand{\Ano}{A_{\rm no}}
\newcommand{\Byes}{B_{\rm yes}}
\newcommand{\Bno}{B_{\rm no}}
\newcommand{\Stat}{\mathsf{Stat}}
\newcommand{\Est}{\mathsf{Est}}
\newcommand{\BP}{\mathsf{BP}}
\newcommand{\PH}{\mathsf{PH}}
\newcommand{\BQ}{\mathsf{BQ}}
\newcommand{\RE}{\mathsf{RE}}
\newcommand{\EXP}{\mathsf{EXP}}
\newcommand{\MIP}{\mathsf{MIP}}
\newcommand{\NP}{\mathsf{NP}}
\newcommand{\PostBQP}{\mathsf{PostBQP}}
\newcommand{\States}{\mathcal{D}}

\title{Unentanglement and Post-Measurement Branching\\in Quantum Interactive Proofs}
\author{Sabee Grewal\thanks{UT Austin. \texttt{\{sabee,kretsch\}@cs.utexas.edu}} 
\and William Kretschmer\footnotemark[1]}

\date{}

\begin{document}

\maketitle

\begin{abstract}
We investigate two resources whose effects on quantum interactive proofs remain poorly understood: the promise of \emph{unentanglement}, and the verifier’s ability to condition on an intermediate measurement, which we call \emph{post-measurement branching}.
We first show that unentanglement can dramatically increase computational power: three-round unentangled quantum interactive proofs equal $\NEXP$, even if only the first message is quantum.
By contrast, we prove that if the verifier uses no post-measurement branching, then the same type of unentangled proof system has at most the power of $\QAM$.
Finally, we investigate post-measurement branching in two-round quantum-classical proof systems.
Unlike the equivalence between public-coin and private-coin classical interactive proofs, we give evidence of a separation in the quantum setting that arises from post-measurement branching.

\end{abstract}

\hypersetup{linktocpage}

\section{Introduction}

Quantum proof systems are one framework for studying quantum mechanical effects on computational complexity. They allow us to ask how uniquely quantum resources---such as superposition, entanglement, and measurement---affect our ability to efficiently verify that certain statements are true. 
A broad goal in this area is understanding which resources increase computational power, which restrict it, and which turn out to be irrelevant.

In this work, we focus on two such resources: the promise of \emph{unentanglement}, and the verifier’s ability to perform \emph{post-measurement branching}. 
At a high level, we show that each of these resources can fundamentally alter the computational power of interactive proof systems.

\subsection{Unentanglement}
Whether entanglement increases or decreases the computational power of quantum proof systems is subtle. 
A striking example is the separation between $\MIP^*$ and $\MIP$, the classes of problems decidable by multiprover interactive proof systems with and without entanglement between provers, respectively.
Whereas $\MIP = \NEXP$~\cite{bfl}, it took considerable additional effort to show that the entangled-prover variant $\MIP^*$ contains $\NEXP$~\cite{ito-vidick2012mipnexp}, because entangled provers could potentially cheat in ways that unentangled provers cannot.
On the contrary, we now know that entanglement adds vastly more power to multiprover interactive protocols: $\MIP^*$ was recently shown to equal $\RE$~\cite{ji2022mipre}, and thus can solve undecidable problems.

Conversely, it is widely believed that in certain proof systems, \textit{unentanglement} affords additional computational power.
This is perhaps best captured by the $\QMA$ vs.\ $\QMA(2)$ question~\cite{kobayashi2001quantumcertificateverificationsingle}, which asks whether two unentangled quantum proofs are more powerful than one.
Despite considerable effort (e.g., \cite{aaronson2008unentanglement,chen2010shortmultiproverquantumproofs,bliertapp,legall2012qma,chiesa2013qma,chailloux2012complexity,HM13,}), nothing beyond the trivial containments $\QMA \subseteq \QMA(2) \subseteq \NEXP$ has been established since $\QMA(2)$'s introduction over two decades ago \cite{kobayashi2001quantumcertificateverificationsingle}.
Moreover, unlike for questions such as $\PTIME$ vs. $\BPP$ or even $\BQP$ vs. $\PH$, there is little consensus in the community about what the answer should be. 

Underlying the $\QMA$ vs.\ $\QMA(2)$ problem is a more basic question: can we exhibit examples where the promise of unentanglement grants quantum proof systems exponentially more computational power?
More recent work has attempted to study this question by defining variants of $\QMA$ and showing them equal to $\NEXP$~\cite{jeronimo2023thepower,bassirian_et_al:LIPIcs.ITCS.2024.9,jeronimo_et_al:LIPIcs.CCC.2024.26,aaronson2024pdqmadqma,Bassirian2025superposition,bassirian2024quantummerlinarthurinternallyseparable}. 
For example, Jeronimo and Wu~\cite{jeronimo2023thepower} introduced the class $\QMA^+(2)$, where the verifier receives two unentangled messages that are additionally promised to consist of real nonnegative amplitudes.
Their main theorem $\QMA^+(2) = \NEXP$ showed that the combination of unentanglement and nonnegative amplitudes can give exponentially more power.
However, this evidence for the power of unentanglement was called into question shortly thereafter, when Bassirian, Fefferman, and Marwaha~\cite{bassirian_et_al:LIPIcs.ITCS.2024.9} showed that nonnegative amplitudes alone, without any promise of unentanglement, give the same power: $\QMA^+ = \NEXP$.

In fact, the only one of these works that showed how unentanglement specifically grants computational power is by Bassirian, Fefferman, Leigh, Marwaha, and Wu~\cite{bassirian2024quantummerlinarthurinternallyseparable}.
There, they define ``internally separable'' variants $\QMA_\mathsf{IS}$ and $\QMA_\mathsf{IS}(2)$ of $\QMA$ and $\QMA(2)$, respectively, in which the quantum witnesses satisfy a certain multipartite entanglement condition.
They then show that $\QMA_\mathsf{IS} \subseteq \EXP$ while $\QMA_\mathsf{IS}(2) = \NEXP$.
(Note that their $\QMA_{\mathsf{IS}}(2) = \NEXP$ is highly sensitive to the completeness and soundness parameters, which cannot (apparently) be generically amplified.)
On the whole, these works identify nontraditional sources of computational power (e.g., non-negative amplitudes and non-collapsing measurements), which have little to do with unentanglement.

In this work, we take a different approach to studying the power of unentanglement.
Rather than considering new computational resources on top of unentanglement, we take a well-studied quantum proof system and investigate how it changes with an additional promise of unentanglement.
We start with $\QIP[3]$, the class of problems decidable by a three-round quantum interactive proof system.
Here the $[3]$ denotes the number of rounds, unlike the $(2)$ in $\QMA(2)$, which refers to the number of unentangled proofs; hence the careful difference in notation.
It is a celebrated result that $\QIP[3] = \QIP = \IP = \PSPACE$~\cite{jain2011qip}, showing that three-round quantum interactive proofs characterize polynomial space.
Because quantum and classical interactive proof systems are so well understood, $\QIP[3]$ serves as a natural baseline for exploring the effects of unentanglement.

We introduce an unentangled variant of $\QIP[3]$, denoted $\uQIP[3]$, that differs from $\QIP[3]$ in exactly one respect: the prover is restricted to apply channels that generate no entanglement between their private workspace qubits and any messages sent to the verifier.
We formalize $\uQIP[3]$ using so-called \textit{entanglement-breaking} channels~\cite{HSR03-entanglement}, which are precisely the set of channels whose action on one half of any bipartite state yields a separable state across the bipartition.
These channels are also known as \textit{measure-and-prepare} channels, because they can be equivalently described as first measuring the input state and then preparing a new quantum state conditioned on the classical measurement outcome. 
Our definition of $\uQIP[3]$ trivially captures $\QMA(2)$ as a special case: the prover sends the first witness in the first round, the verifier does nothing in the second round, and the prover sends the second witness (unentangled with the first) in the third and final round.

Our first result establishes that the promise of unentanglement in $\uQIP[3]$ yields dramatically greater power than $\QIP[3] = \PSPACE$:

\begin{theorem}\label{thm:intro-uqip-nexp}
    $\uQIP[3] = \NEXP$.
\end{theorem}

Note that, unlike prior comparable work~\cite{jeronimo2023thepower,bassirian_et_al:LIPIcs.ITCS.2024.9,bassirian2024quantummerlinarthurinternallyseparable}, our result is insensitive to the precise completeness-soundness gap.
In particular, \cref{thm:intro-uqip-nexp} holds for any constant gap less than $1$.

For the upper bound, we argue that an $\NEXP$ machine can nondeterministically guess an exponentially-large description of the prover's strategy and then verify whether it causes the verifier to accept with high probability or not.
To prove the complementary lower bound, we make use of a certain quantum PCP for $\NEXP$ introduced by Raz some 20 years ago~\cite{Raz05-qipcp}.
Raz showed that any $\NEXP$ language admits a polynomial-time quantum verifier that receives two inputs: a polynomial-length quantum witness, and an exponentially-large classical proof (readable by query access).
The verifier measures the quantum witness and, based on the measurement outcome, queries a \textit{single} polynomial-size block of the proof.
We argue that this PCP can be simulated by a $\uQIP[3]$ protocol: the prover sends the quantum witness in the first round, the verifier measures it and sends the resulting classical query to the prover in the second round, and the prover responds with the answer to the classical query in the final round.

One can view \Cref{thm:intro-uqip-nexp} as popularizing and modernizing Raz's quantum PCP result in the context of $\QMA(2)$ and the power of unentanglement, as the proof is not particularly difficult with Raz's result in hand.
Indeed, Raz's discussion briefly mentions something close to \Cref{thm:intro-uqip-nexp}: that an interactive prover with quantum power in the first round and classical power thereafter can convince a verifier of the solution to an $\NEXP$ problem~\cite[Section 1.5]{Raz05-qipcp}.
Nevertheless, we found that \Cref{thm:intro-uqip-nexp} surprised every expert in the field with whom we consulted.

It is striking that the containment of $\NEXP$ in $\uQIP[3]$ makes only partial use of the proof system's quantum capabilities, as the second and third messages are purely classical.
Naturally, one might wonder why we need the promise of unentanglement at all: if the verifier knows that the final message is classical, then doesn't that already guarantee zero entanglement between the prover's first and last messages?
The key observation is that the source of power is not unentanglement between the messages, but rather from the unentanglement between the \textit{prover's workspace qubits} and the messages.

The following simple example illuminates the situation: consider a $\QIP[3]$ protocol in which the verifier challenges the prover to a version of the CHSH game~\cite{CHTW04-chsh}, with the verifier playing both the role of the referee and one of the players.
In the first round, the prover sends the verifier a single qubit.
The verifier then uniformly samples $x \sim \{0,1\}$ and sends it to the prover, who responds with a single classical bit $a$.
Next, the verifier uniformly samples $y \sim \{0,1\}$.
If $y = 0$, the verifier measures the qubit sent by the prover in the $\{\ket 0, \ket 1\}$ basis; else they measure in the $\{\ket{+},\ket{-}\}$ basis.
Calling the measurement outcome $b$, the verifier accepts if and only if $xy = a \oplus b$.
An entangled prover can succeed with probability $\cos^2(\pi/8) \approx 0.85$ by sending one half of a Bell pair in the first round, and, in the final round, measuring the other half according to the optimal CHSH strategy. By contrast, an \emph{unentangled} prover can make the verifier accept with probability at most $0.75$.
Hence, even with only a single quantum message in the first round, the promise of unentanglement places a nontrivial restriction on the set of valid $\QIP[3]$ prover strategies.

\subsection{Post-Measurement Branching and Unentanglement}

One key difference between $\uQIP[3]$ and $\QMA(2)$ is adaptivity: a $\uQIP[3]$ verifier can condition their round-2 challenge to the prover on the result of a measurement, possibly applied to the round-1 message.
By contrast, a $\QMA(2)$ verifier receives a pair of unentangled witnesses simultaneously, without the ability for either witness to depend on a chosen challenge.
In the interest of isolating the source of $\uQIP[3]$'s exponential power, it is natural to ask whether the quantum-classical-classical unentangled proof system for $\NEXP$ necessitated this sort of adaptivity.
For example, if instead the verifier simply sent the prover random coin tosses in round $2$, could they still verify solutions to $\NEXP$ problems?

In the context of quantum proof systems, we refer to this type of adaptivity as \emph{post-measurement branching}, which means the ability to condition on a partial measurement of a state while retaining the residual post-measurement state.
Our proposal to study the same quantum proof system with and without post-measurement branching mirrors the difference between $\AM[k]$ and $\IP[k]$: in $\AM[k]$ the verifier's messages to the prover consist of public coin tosses, whereas in $\IP[k]$ the messages can be arbitrary polynomial-time randomized computations on the prior transcript of the protocol.
Classically, we know that adaptivity cannot help much: $\IP[k] \subseteq \AM[k+2]$~\cite{goldwasser1986publicvsprivate}, and for constant $k$, $\AM[k] \subseteq \AM[2]$~\cite{babaicollapse,BM88-am}.
But should we expect the analogous equivalence to hold for quantum protocols, unentangled or otherwise?

Concretely, consider the subclass of $\uQIP[3]$ in which the round-2 message consists of public coin tosses and the round-3 message is classical.
We call this subclass $\uQMACM$ because it behaves like $\QMAM$~\cite{marriott2005quantum}, except that the prover's private and message registers are always unentangled and the last message is classical.
Then does $\uQMACM = \uQIP[3] = \NEXP$?
Our second result gives strong evidence that the answer is no:

\begin{theorem}
\label{thm:qmacm_in_qam_intro}
    $\uQMACM \subseteq \QAM$.
\end{theorem}

Here, $\QAM$ is the set of problems verifiable by an interaction in which the verifier (Arthur) sends public coin tosses and the prover (Merlin) responds with a quantum message~\cite{marriott2005quantum}.
In contrast to $\uQIP[3] = \NEXP$, where restricting to unentangled provers significantly \textit{increased} computational power compared to $\QIP[3] = \PSPACE$, here the unentangled variant $\uQMACM$ is quite plausibly \emph{weaker} than its entangled variant $\QMACM$.
Indeed, whereas $\uQMACM \subseteq \QAM = \BP \cdot \QMA \subseteq \BPP^\PP$ (where $\BP \cdot \QMA$ denotes problems that have a randomized many-one reduction to $\QMA$), we know of no better upper bound on the corresponding entangled proof system than $\QMACM \subseteq \QMAM = \QIP[3] = \PSPACE$~\cite{jain2011qip}.
$\uQMACM$ is possibly equal to $\QMA$, and in fact the classes coincide with polynomial-size advice: $\QMApoly = \QAMpoly = \uQMACMpoly$, because $\QAM = \BP \cdot \QMA$ and the $\BP \cdot$ operator can be derandomized with advice (cf.~\cite{aaronson2006qma,aaronson2023certified}).
Thus, \Cref{thm:qmacm_in_qam_intro} illustrates both the necessity of post-measurement branching for making certain proof systems equal to $\NEXP$, and the surprising fact that unentanglement may hinder the power of an interactive proof system.

The proof of \Cref{thm:qmacm_in_qam_intro} involves simulating the $\uQMACM$ proof system in two rounds by combining Merlin's first and last messages into one.
In the $\QAM$ protocol, Arthur first sends Merlin polynomially many independent challenges that he could have sent in round $2$ of the $\uQMACM$ protocol.
Then Arthur asks for both Merlin's quantum proof that he would have sent in round $1$, and Merlin's classical answers that he would have given in round 3 in response to each of the challenges.
We argue that the soundness of the protocol is approximately preserved if Arthur runs his original $\uQMACM$ check on a random one of the challenges.
This strategy is somewhat analogous to the $\AM[k] \subseteq \AM[2]$ collapse theorem~\cite{BM88-am}, but requires heavier tools to handle the quantum part of the message.
For example, a crucial ingredient in our proof comes from one-way communication complexity: any $n$-qubit quantum state $\rho$ can be ``compressed'' into a $\poly(n)$-bit message, from which the expectation of $\exp(n)$ different measurements on $\rho$ may be later estimated~\cite{aaronson2005limitations}.
We do not directly use this compression scheme in the $\QAM$ protocol, but it indirectly allows us to apply a union bound over the set of $n$-qubit states as if there were only $2^{\poly(n)}$ of them, instead of $2^{\exp(n)}$.

\subsection{Post-Measurement Branching with Classical Messages}
Finally, we turn to the simplest setting in which the effect of post-measurement branching can be studied: two-round interactive proofs with classical messages and a quantum verifier. 
Specifically, we study the complexity classes $\QCAM$ and $\QCIP[2]$. In $\QCAM$, the verifier's sole message consists of random coin tosses. In $\QCIP[2]$, by contrast, the verifier may send an arbitrary classical message generated through a partial measurement of a quantum state; both the classical outcome and the corresponding post-measurement state can then be used later in the verification procedure. 

Recall that $\AM$ and $\IP[2]$ coincide, and more generally, constant-round public-coin ($\AM$) and private-coin ($\IP$) protocols have the same computational power~\cite{babaicollapse,goldwasser1986publicvsprivate,BM88-am}. 
In contrast, we show that the quantum setting exhibits an apparent separation: $\QCIP[2]$ potentially has greater power than $\QCAM$. 
Our findings are summarized in the following theorem.

\begin{theorem}
\label{thm:qcam_containments_intro}
 $\QCAM = \BP \cdot \QCMA \subseteq \BQ \cdot \QCMA \subseteq \QCIP[2] \subseteq \BQP^{\NP{^\PP}}$.
\end{theorem}

$\BP \cdot \QCMA$ (resp.\ $\BQ \cdot \QCMA$) is the class of promise problems that admit a randomized (resp.\ quantum) many-one reduction to a promise problem in $\QCMA$. 
We note that $\QCAM = \BP \cdot \QCMA$ was originally proven by Marriott \cite{marriott2003nondeterminism}, but we include a complete proof in \cref{sec:qcip-qcam} in more standard notation.

The containments involving $\BP \cdot \QCMA$ and $\BQ \cdot \QCMA$ are reasonably straightforward applications of the definitions.
Placing an upper bound on $\QCIP[2]$, however, takes more effort.
Intuitively, it works as follows: first, use the base $\BQP$ machine to generate the verifier's round-$1$ message.
Then, we will use the $\NP^\PP$ machine to simulate the prover.
The idea is to nondeterministically guess the prover's round-$2$ message and then verify using the $\PP$ oracle whether that message would be accepted by the verifier.
$\PP$ suffices for this step because of Aaronson's $\PP = \PostBQP$ theorem~\cite{aaronson2005quantum}, which shows that $\PP$ equals the set of problems decidable by an efficient quantum machine with postselection.
Using postselection, one can condition on producing the same message that the verifier sampled in round $1$, resulting in the same residual state that the verifier uses to decide acceptance or rejection at the end.

Unlike the classical equivalence $\IP[2] = \AM$, \Cref{thm:qcam_containments_intro} hints that $\QCIP[2]$ is more powerful than $\QCAM$, because the set of problems quantumly reducible to $\QCMA$ is plausibly larger than the set classically reducible to $\QCMA$.
However, this distinction alone does not make full use of $\QCIP[2]$'s extra power, as $\BQ \cdot \QCMA$ is a class that uses no post-measurement branching.
Looking at the higher end of the containments, we found it considerably more challenging to place an upper bound on $\QCIP[2]$ than $\QCAM$ \textit{precisely because} the former uses post-measurement branching, and thus simulating Merlin requires a handle on the post-measurement state.

\subsection{Open Problems}

We conclude with some directions for future work.

\begin{enumerate}
    \item We proved that $\uQIP[3] = \NEXP$. Are there quantum proof systems between $\QMA(2)$ and $\uQIP[3]$ that capture $\NEXP$? An unentangled version of $\QMAM$ is a natural candidate to study.
    
    \item A key difference between $\QCAM$ and $\QCIP[2]$ (and even between $\BQ \cdot \QCMA$ and $\QCIP[2]$) is the power of post-measurement branching. What more can be said about this power? For instance, can one show stronger containments than $\BQ \cdot \QCMA \subseteq \QCIP[2]$? 
    
    \item It is known that $\IP[k] = \IP[2] = \AM$ for every constant $k \geq 2$~\cite{babaicollapse,goldwasser1986publicvsprivate,BM88-am}. Is an analogous collapse true for $\QCIP[k]$?  
    
    \item Are there oracles relative to which any of the containments in \Cref{thm:qcam_containments_intro} are strict?
\end{enumerate}

\section{Unentangled Quantum Interactive Proofs}\label{sec:uqip3}
In this section, we introduce an unentangled three-round quantum interactive proof system, denoted $\uQIP[3]$. 
Our definition mirrors $\QIP[3]$, except that Merlin's actions are restricted to entanglement-breaking channels.
After defining our model, we prove that $\uQIP[3] = \NEXP$.
The result adapts a certain type of quantum PCP for $\NEXP$ introduced by Raz~\cite{Raz05-qipcp} (quoted below in \Cref{lem:raz_nexp_pcp}).

We begin by recalling the definition of $\QIP[3]$.
A $\QIP[3]$ verification procedure is specified by a polynomial-time uniformly generated family of quantum circuits $V = \{V_1^x, V_2^x : x \in \{0,1\}^*\}$. 
On an input $x$ of length $n$, these circuits determine the actions of the verifier across the three-message interaction.
Each circuit acts on $\poly(n)$-sized registers partitioned into \emph{message} qubits $\mathcal{M}$, exchanged with the prover, and \emph{verifier workspace} qubits $\mathcal{V}$, retained by the verifier throughout.

The prover is an unrestricted family $P = \{P_1^x, P_2^x : x \in \{0,1\}^*\}$ of arbitrary quantum operations that likewise act on the same message qubits $\mathcal{M}$ and \emph{prover workspace} qubits $\mathcal{P}$ (which need not be polynomially-bounded in size). 
The interaction proceeds as follows:
\begin{enumerate}
    \item The three registers $\mathcal{V}$, $\mathcal{M}$, $\mathcal{P}$ are each initialized to the all-zeros state.
    \item The prover applies $P_1^x$ to $\mathcal{P}$ and $\mathcal{M}$.
    \item The verifier applies $V_1^x$ to $\mathcal{V}$ and $\mathcal{M}$.
    \item The prover applies $P_2^x$ to $\mathcal{P}$ and $\mathcal{M}$.
    \item  Finally, the verifier applies $V_2^x$ to $\mathcal{V}$ and $\mathcal{M}$ and measures a designated output qubit to decide acceptance or rejection.
\end{enumerate}

Sometimes the prover and verifier are called ``Merlin'' and ``Arthur'' respectively.
We will typically only use these names in interactive protocols where the verifier's messages consist of public coin tosses, consistent with the distinction between the complexity classes $\AM[k]$ and $\IP[k]$.

We now formally define the class $\QIP[3]$. 

\begin{definition}[{$\QIP[3]$}]
\label{def:QIP[3]}
A promise problem $A = (\Ayes, \Ano)$ is in $\QIP[3, c, s]$ for polynomial-time computable functions $c,s : \N \to [0,1]$ if there exists a $\QIP[3]$ verification procedure $V$ with the following properties: 
    \begin{itemize}
        \item \emph{Completeness.} For all $x \in \Ayes$, there exists a quantum prover $P$ that causes $V$ to accept $x$ with probability at least $c(\abs{x})$.
        \item \emph{Soundness.} For all $x \in \Ano$, every quantum prover $P$ causes $V$ to accept $x$ with probability at most $s(\abs{x})$.
    \end{itemize}
We define $\QIP[3] \coloneqq \QIP[3, 2/3, 1/3]$.
\end{definition}

It is known that any polynomial-round quantum interaction can be parallelized to three rounds \cite{kitaev2000parallelization} and that $\QIP[3, 2/3, 1/3] = \QIP[3, 1, 2^{-\poly}] = \PSPACE$ \cite{jain2011qip}.  

\begin{remark}[Entanglement between registers]
A key feature of $\QIP[3]$---one that is central to this work---is that the private workspaces of both the prover and verifier may be entangled with the message registers exchanged during the interaction.
\end{remark}

We turn to defining our unentangled variant of $\QIP[3]$, which adds an additional restriction on the prover involving \textit{entanglement-breaking} channels~\cite{HSR03-entanglement}.
These channels are so-called because they are precisely the set of channels $\Phi$ with the property that for any input density matrix $\rho$, $(I \otimes \Phi)(\rho)$ is separable (across the cut between the output of $\Phi$ and the tensored identity factor).
An equivalent characterization of an entanglement-breaking channel $\Phi$ is one that takes the form
\[
  \Phi(\rho) \;=\; \sum_{\ell} \tr(E_\ell \rho)\, \ket{\phi_\ell}\!\bra{\phi_\ell}.
\]
for some POVM $\{E_\ell\}$ and set of states $\ket{\phi_\ell}$.
Operationally, this means that an entanglement-breaking channel applies a measurement to the input state and prepares a new state conditioned on the classical outcome of the measurement.
For this reason, entanglement-breaking channels are sometimes called \textit{measure-and-prepare} channels.

An \textit{unentangled} $\QIP[3]$ prover is a family $P = \{P_1^x, P_2^x : x \in \{0,1\}^*\}$ of quantum operations that likewise act on $\mathcal{P}$ and $\mathcal{M}$, subject to the restriction that each $P_i^x$ is the sequential composition of:
\begin{enumerate}
    \item Applying some arbitrary channel $\Lambda_i^x$ to $\mathcal{P}$ and $\mathcal{M}$;
    \item Applying an entanglement-breaking channel $\Phi_i^x$ whose input qubits are $\mathcal{S} \cup \mathcal{M}$ and output qubits are $\mathcal{M}$, for some $\mathcal{S} \subseteq \mathcal{P}$;
    \item Reinitializing the qubits in $\mathcal{S}$ to $\ket{0}$. (This step is only needed to preserve the size of $\mathcal{P}$.)
\end{enumerate}

Use of the entanglement-breaking channel $\Phi_i^x$ ensures that after completion of the prover's operation $P_i^x$, the message register $\mathcal{M}$ is unentangled from the prover's workspace qubits $\mathcal{P}$.

\begin{figure}
    \centering
    \[
    \Qcircuit @C=1em @R=1em @!R {
    &&& \raisebox{-2em}{$P_1^x$} &&&&&& \raisebox{-2em}{$P_2^x$}
        \\
    \lstick{\mathcal{P} \setminus \mathcal{S}\quad\ket{0}} 
        & {/}\qw
        & \multigate{2}{\Lambda_1^x}
        & \qw
        & \qw
        & \qw
        & \qw
        & \qw
        & \multigate{2}{\Lambda_2^x}
        & \qw
        & \qw
        & \qw
        & \qw
        & \qw
        & \qw
        & \qw
        \\
    \lstick{\mathcal{S}\quad\ket{0}} 
        & {/}\qw
        & \ghost{\Lambda_1^x}
        & \multimeasureD{1}{\{E_{1,\ell}^x\}}
        & 
        & \ket{0}
        &
        & \qw
        & \ghost{\Lambda_2^x}
        & \multimeasureD{1}{\{E_{2,\ell}^x\}}
        & 
        & \ket{0}
        &
        & \qw
        & \qw
        & \qw
        \\
    \lstick{\mathcal{M}\quad\ket{0}} 
        & {/}\qw
        & \ghost{\Lambda_1^x}
        & \ghost{\{E_{1,\ell}^x\}}
        & \cw
        & \ket{\phi_{1,\ell}}
        &
        & \multigate{1}{V_1^x}
        & \ghost{\Lambda_2^x}
        & \ghost{\{E_{2,\ell}^x\}}
        & \cw
        & \ket{\phi_{2,\ell}}
        &
        & \qw
        & \multigate{1}{V_2^x}
        & \qw
        \\
    \lstick{\mathcal{V}\quad\ket{0}} 
        & {/}\qw
        & \qw
        & \qw
        & \qw
        & \qw
        & \qw
        & \ghost{V_1^x}
        & \qw
        & \qw
        & \qw
        & \qw
        & \qw
        & \qw
        & \ghost{V_2^x}
        & \meter
    \gategroup{2}{3}{4}{7}{.7em}{--}
    \gategroup{2}{9}{4}{13}{.7em}{--}
    }
    \]
    \caption{\label{fig:uQIP[3]_general} The general form of a $\uQIP[3]$ interaction between an unentangled prover $P = \{P_1^x, P_2^x\}$ and verifier $V = \{V_1^x, V_2^x\}$.}
\end{figure}
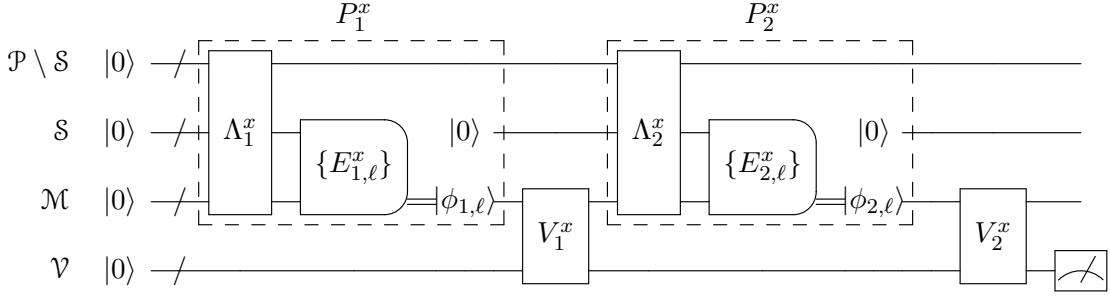

\Cref{fig:uQIP[3]_general} depicts the interaction between an unentangled prover and verifier.
This type of interaction underlies $\uQIP[3]$, whose formal definition replaces the prover in $\QIP[3]$ with an unentangled prover.

\begin{definition}[{$\uQIP[3]$}]
A promise problem $A=(A_{\mathrm{yes}},A_{\mathrm{no}})$ is in $\uQIP[3,c,s]$ for polynomial-time computable $c,s:\N\to[0,1]$ if there exists a $\QIP[3]$ verification procedure $V$ with the following properties:
\begin{itemize}
  \item \emph{Completeness.} For all $x\in A_{\mathrm{yes}}$, there exists an \textit{unentangled} prover $P$ that makes $V$ accept with probability at least $c(|x|)$.
  \item \emph{Soundness.} For all $x\in A_{\mathrm{no}}$, every \textit{unentangled} prover $P$ makes $V$ accept with probability at most $s(|x|)$.
\end{itemize}
We define $\uQIP[3] \coloneqq \uQIP[3, 2/3, 1/3]$.
\end{definition}

When defining such complexity classes, one should always ask whether the the choices of completeness $2/3$ and soundness $1/3$ are arbitrary---in particular, can they be amplified, say by parallel repetition?
It will follow from our results that the completeness and soundness parameters can be amplified to $1$ and $o(1)$, respectively (\Cref{cor:uqip_amplification}).

Notice that $\uQIP[3]$ places no restrictions on the verifier's ability to send entangled messages to the prover; the promise merely guarantees that the prover never maintains entanglement with their sent messages.

\subsection{Upper Bounding \texorpdfstring{$\uQIP[3]$}{Unentangled QIP}}

To place an upper bound on $\uQIP[3]$, we first make note of some useful properties of entanglement-breaking channels that allow us to reduce the complexity of the prover.
This first lemma shows that one can straightforwardly upper bound the description complexity of an entanglement-breaking channel, which is \textit{a priori} infinite.

\begin{lemma}
    \label{lem:povm_size}
    Suppose $\Phi$ is an entanglement-breaking channel from $m$ qubits to $n$ qubits.
    Then $\Phi$ admits a decomposition in terms of a POVM $\{E_\ell\}$ and a set of pure states $\{\ket{\phi_\ell}\}$ with at most $4^{m+n}$ terms:
    \[
    \Phi(\rho) \;=\; \sum_{\ell = 1}^{4^{m + n}} \tr(E_\ell \rho)\, \ket{\phi_\ell}\!\bra{\phi_\ell}.
    \]
\end{lemma}

\begin{proof}
    The proof follows \cite[Theorem 4]{HSR03-entanglement} exactly, with the sole addition of some extra accounting.
    \cite[Theorem 4]{HSR03-entanglement} shows that if $\Phi$ is entanglement-breaking, then its Choi state
    \[
    (I \otimes \Phi)(\ket{\beta}\!\bra{\beta})
    \]
    is separable (i.e., a mixture of product states), where $\ket{\beta} = \frac{1}{\sqrt{2^m}} \sum_{j \in \{0,1\}^m} \ket{j}\ket{j}$.
    By a result of Horodecki~\cite{Hor97-separability}, every separable state on $m + n$ qubits is a convex combination of at most $4^{m+n}$ pure product states, and therefore the Choi state admits a decomposition:
    \[
    (I \otimes \Phi)(\ket{\beta}\!\bra{\beta}) = \sum_{\ell = 1}^{4^{m+n}} p_\ell \ket{v_\ell}\!\bra{v_\ell} \otimes \ket{w_\ell}\!\bra{w_\ell},
    \]
    where $\{p_\ell\}$ are probabilities summing to $1$ and $\{\ket{w_\ell}\}$, $\{\ket{v_\ell}\}$ are lists of normalized pure states.

    Now let $\Omega$ be the map
    \[
    \Omega(\rho) \coloneqq \sum_{\ell = 1}^{4^{m+n}} \tr(d p_\ell \ket{v_\ell}\!\bra{v_\ell}\rho) \ket{w_\ell}\!\bra{w_\ell},
    \]
    which has the form required by the lemma.
    Using $\ket{v_\ell} = \sum_j \ket{j} \braket{j}{v_\ell}$, one easily verifies that
    \begin{align*}
    (I \otimes \Omega)(\ket{\beta}\!\bra{\beta})
    &= \sum_{jk\ell}\ket{j}\!\bra{k} \otimes \ket{w_\ell}\!\bra{w_\ell} p_\ell \braket{j}{v_\ell}\braket{v_\ell}{k}\\
    &= \sum_{\ell} p_\ell \ket{v_\ell}\!\bra{v_\ell} \otimes \ket{w_\ell}\!\bra{w_\ell}\\
    &= (I \otimes \Phi)(\ket{\beta}\!\bra{\beta}).
    \end{align*}
    Two channels are equal if and only if their Choi states are the same, so $\Phi = \Omega$.

    To complete the proof, one must verify that $\{dp_\ell \ket{v_\ell}\!\bra{v_\ell}\}$ is a POVM.
    This follows by taking a partial trace of the Choi state:
    \begin{align*}
        \tr_2 \left[(I \otimes \Phi)(\ket{\beta}\!\bra{\beta})\right]
        &= \frac{I}{d}\\
        &= \sum_{\ell = 1}^{4^{m+n}} p_\ell \ket{v_\ell}\!\bra{v_\ell}.\qedhere
    \end{align*}
\end{proof}

Next, we argue that the prover can further simplify their strategy by eliminating the use of private qubits:

\begin{lemma}
    \label{lem:prover_canonical}
    Suppose there exists a $\uQIP[3]$ prover $P$ that makes $V$ accept with probability at least $p(|x|)$.
    Then there exists a prover $\overline{P}$ that also makes $V$ accept with probability at least $p(|x|)$, but for which
    \begin{enumerate}
        \item\label{item:pure_first_message} The output of $\overline{P}_1^{x}$ is a pure state on $\mathcal{M}$,
        \item\label{item:no_private} $\overline{P}$ uses no prover workspace qubits, and
        \item\label{item:entanglement_breaking} $\overline{P}_1^{x}$ and $\overline{P}_2^{x}$ are themselves entanglement-breaking channels.
    \end{enumerate}
\end{lemma}

\begin{proof}
    After the prover applies $P_1^x$ to $\ket{0}_{\mathcal{P}\mathcal{M}}$ the state of registers $\mathcal{P}$ and $\mathcal{M}$ has the form
    \[
    \sum_\ell p_\ell \; \sigma_{\ell, \mathcal{P}} \otimes \ket{\psi_\ell}\!\bra{\psi_\ell}_{\mathcal{M}},
    \]
    for some probabilities $\{p_\ell\}$, mixed states $\{\sigma_\ell\}$, and pure states $\{\ket{\psi_\ell}\}$ parameterized by the possible outcomes $\ell$ of the POVM underlying the entanglement-breaking channel $\Phi_1^x$.
    Consider postselecting on a particular outcome $\ell$.
    By convexity, there must exist a choice $\ell = \ell^*$ such that replacing $P_1^x$ with direct preparation of $\sigma_{\ell^*, \mathcal{P}} \otimes \ket{\psi_{\ell^*}}\!\bra{\psi_{\ell^*}}_{\mathcal{M}}$ causes the interaction between prover and verifier to accept with probability at least $p(|x|)$.
    Let $P'$ be the prover derived from $P$ by making this replacement, which causes it to satisfy \Cref{item:pure_first_message}.
    To avoid confusion with notation later in the proof, we drop the $\ell^*$ subscript and call the initial state simply $\sigma_{\mathcal{P}} \otimes \ket{\psi}\!\bra{\psi}_{\mathcal{M}}$.

    Next, we observe that one can eliminate the need to prepare $\sigma_{\mathcal{P}}$ on a second register, and thus remove the private workspace qubits.
    Let $\{E_\ell\}$ be the POVM and $\{\ket{\phi_\ell}\}$ be the set of states underlying $\Phi_2^x$, for which
    \[
      \Phi_2^x(\rho_{\mathcal{PM}}) \;=\; \sum_{\ell} \tr(E_\ell \rho_{\mathcal{PM}})\, \ket{\phi_\ell}\!\bra{\phi_\ell}_{\mathcal{M}}.
    \]
    Now define $\overline{P}$ by
    \[
    \overline{P}^x_1(\rho_{\mathcal{M}}) \;\coloneqq\; \tr(\rho_\mathcal{M}) \ket{\psi}\!\bra{\psi}_\mathcal{M}
    \]
    and
    \[
        \overline{P}_2^x(\rho_\mathcal{M}) \;\coloneqq\; \sum_\ell \tr(E_\ell \Lambda_2^x(\sigma_\mathcal{P} \otimes \rho_\mathcal{M})) \ket{\phi_\ell}\!\bra{\phi_\ell}.
    \]
    Then clearly $\overline{P}$ satisfies \Cref{item:no_private}, since $\overline{P}^x_1$ and $\overline{P}^x_2$ both map $\mathcal{M}$ to $\mathcal{M}$.
    Additionally, the interaction between $\overline{P}$ and $V$ has the same acceptance probability as that between $P'$ and $V$, because we essentially deferred initializing $\sigma_\mathcal{P}$ until $\overline{P}_2^x$, and then traced out $\mathcal{P}$ afterwards.

    The last condition we must verify is that $\overline{P}_1^{x}$ and $\overline{P}_2^{x}$ are entanglement-breaking channels.
    This is immediate for $\overline{P}_1^{x}$.
    To see that $\overline{P}_2^{x}$ is entanglement-breaking, first let $\Psi_2^x(\rho_{\mathcal{M}}) \coloneqq \Lambda_2^x(\sigma_\mathcal{P} \otimes \rho_\mathcal{M})$.
    Then
    \begin{align*}
        \overline{P}_2^x(\rho_\mathcal{M}) \;&=\; \sum_\ell \tr(E_\ell \Psi_2^x(\rho_\mathcal{M})) \ket{\phi_\ell}\!\bra{\phi_\ell}\\
        &=\; \sum_\ell \tr(\Psi_2^{x*}(E_\ell) \rho_\mathcal{M}) \ket{\phi_\ell}\!\bra{\phi_\ell},
    \end{align*}
    where $\Psi_2^{x*}$ is the channel adjoint of $\Psi_2^{x}$.
    This satisfies the definition of entanglement breaking because $\Psi_2^{x*}(E_\ell)$ is a POVM, since the adjoint of a CPTP map is completely positive and unital.
\end{proof}

\begin{figure}
    \centering
    \[
    \Qcircuit @C=1em @R=1em @!R {
    \lstick{\raisebox{-1.5em}{$P_1^x$}\ }&&& \raisebox{-1.5em}{\qquad\quad $P_2^x$}
        \\
    \lstick{\mathcal{M}\ \ 
    \tikz[baseline=(X.base)] \node[draw, dashed, dash pattern=on 5pt off 5pt, inner sep=3pt](X){$\ket{\psi}$};
    } 
        & {/}\qw
        & \multigate{1}{V_1^x}
        & \measureD{\{E_\ell\}}
        & \cw
        & \ket{\phi_\ell}
        & 
        & \multigate{1}{V_2^x}
        & \qw
        \\
    \lstick{\mathcal{V}\ \quad\ket{0}} 
        & {/}\qw
        & \ghost{V_1^x}
        & \qw
        & \qw
        & \qw
        & \qw
        & \ghost{V_2^x}
        & \meter
        \gategroup{2}{4}{2}{7}{.7em}{--}
    }
    \]
    \caption{\label{fig:uQIP[3]-simplified} The simplified form of $\uQIP[3]$ interaction between an unentangled prover $P = \{P_1^x, P_2^x\}$, in canonical form that derives from \Cref{lem:prover_canonical}, and verifier $V = \{V_1^x, V_2^x\}$.}
\end{figure}
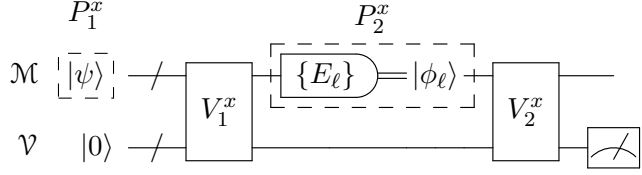

\Cref{fig:uQIP[3]-simplified} shows the canonical form of an unentangled prover that derives from \Cref{lem:prover_canonical}
This simplification lets us simulate $\uQIP[3]$ in nondeterministic exponential time.

\begin{theorem}
    \label{thm:uQIP-in-NEXP}
    For any $c(n) - s(n) \ge \frac{1}{2^{\poly(n)}}$, $\uQIP[3, c, s] \subseteq \NEXP$.
\end{theorem}

\begin{proof}
    Given $A \in \uQIP[3]$, let $V$ be a corresponding $\uQIP[3]$ verification procedure with completeness $c$ and soundness $s$.
    On input $x \in \{0,1\}^n$, to decide whether $x \in \Ayes$ or $x \in \Ano$, the $\NEXP$ procedure is to nondeterministically guess classical descriptions of the entanglement-breaking channels $P_1^n, P_2^n$ that act on $\mathcal{M}$ (up to some small error in diamond norm, say $\frac{c(|x|) - s(|x|)}{100}$, which requires $\poly(n)$ bits of precision), and then verify in exponential time whether the interaction between $P_1^n, P_2^n$ and $V_1^n, V_2^n$ causes the verifier to accept with probability at least $\frac{c(|x|) + s(|x|)}{2}$.
    Note crucially by \Cref{lem:povm_size} that the descriptions of $P_1^n$ and $P_2^n$ require at most $2^{\poly(n)}$ bits, as the $\NEXP$ machine can guess the lists of POVM elements and pure states that describe the channels.
    For any optimal prover, restricting attention to entanglement-breaking $P_1^n, P_2^n$ with no prover workspace qubits is without loss of generality, by \Cref{lem:prover_canonical}.
    As long as the total error incurred by the finite-precision approximation is less than $c(|x|) - s(|x|)$, the completeness/soundness guarantee of the $\uQIP[3]$ verifier ensures that ``yes'' instances have an accepting witness to the $\NEXP$ verifier, while ``no'' instances do not.
\end{proof}

One could analogously define $\uQIP[k]$ for any polynomially-bounded $k$, and hope to show the same containment in $\mathsf{NEXP}$.
Unfortunately, it appears that \Cref{lem:prover_canonical} does not directly generalize beyond three rounds in showing that workspace qubits are superfluous.
For example, one could imagine a scenario in which the prover sends the verifier one half of an EPR pair in round $2$, and then in rounds $k-1$ and $k$ plays a sort of CHSH game with the prover.
This would require the prover to retain their half of the EPR pair until round $k$.

\subsection{Lower bounding \texorpdfstring{$\uQIP[3]$}{Unentangled QIP[3]}}

In this section, we prove a complementary lower bound on $\uQIP[3]$, implying that $\uQIP[3] = \NEXP$.
The proof relies on the following quantum PCP for $\NEXP$, due to Raz~\cite{Raz05-qipcp}.

\begin{lemma}
    \label{lem:raz_nexp_pcp}
    For any language $L \in \NEXP$, there exists a polynomial-time quantum oracle algorithm $Q^{(\cdot)}$ that makes a single \emph{classical}\footnote{Meaning, the verifier measures a register to obtain $y \in \{0,1\}^{\poly(n)}$, and then queries the value of $f(y)$.} query to the oracle such that
    \begin{itemize}
    \item \emph{Completeness.} For every $x\in L$, there exists a state $\ket{\psi}$ on $\poly(|x|)$ qubits and an oracle $f: \{0,1\}^{\poly(|x|)} \to \{0,1\}^{\poly(|x|)}$ for which $Q^f(x, \ket{\psi})$ accepts with probability $1$.
    \item \emph{Soundness.} For every $x\not\in L$, for every state $\ket{\psi}$ on $\poly(|x|)$ qubits and oracle $f: \{0,1\}^{\poly(|x|)} \to \{0,1\}^{\poly(|x|)}$, $Q^f(x, \ket{\psi})$ accepts with probability $o(1)$ (as a function of $|x|$).
    \end{itemize}
\end{lemma}

The $\uQIP[3]$ containment of $\NEXP$ amounts to a direct simulation of this PCP.

\begin{theorem}\label{thm:nexp-in-uqip}
   $\NEXP \subseteq \uQIP[3, 1, o(1)]$.
\end{theorem}

\begin{proof}
    Consider a $\uQIP[3]$ verifier in which $V_1^x$ initially treats the register $\mathcal{M}$ as the input $\ket{\psi}$ to the algorithm $Q^{(.)}$ from \Cref{lem:raz_nexp_pcp}, then performs the intermediate measurement of $Q^{(.)}(x, \ket{\psi})$ to obtain the input $y \in \{0,1\}^{\poly(n)}$ to the classical query, and finally sends $y$ over $\mathcal{M}$ to the prover.
    In the final round, $V_2^x$ views $\mathcal{M}$ as the classical response to the oracle query $f(y)$, measures $\mathcal{M}$ in the computational basis, and then performs the final measurement of $Q$ to decide acceptance or rejection.

    This protocol has completeness $1$, as witnessed by the unentangled prover $P$ for which $P_1^x$ prepares the state $\ket{\psi}$ and $P_2^x$ evaluates the function $f$ that causes $Q^f(x,\ket{\psi})$ to accept with probability $1$.
    For soundness, after putting the prover in the canonical form of \Cref{lem:prover_canonical}, notice that we can further simplify the description of the prover by assuming without loss of generality that $P_2^x$ evaluates some classical function, because the verifier measures $\mathcal{M}$ in the computational basis both at the end of $P_1^x$ and at the start of $P_2^x$.
    Thus, such a prover strategy is fully specified by the state $\ket{\psi}$ sent at the start and the classical function $f: \{0,1\}^{\poly(n)} \to \{0,1\}^{\poly(n)}$ applied at the end.
    It follows that the $\uQIP[3]$ soundness matches that of the underlying quantum PCP for $\NEXP$, which \Cref{lem:prover_canonical} shows to be $o(1)$.
\end{proof}

Combining \Cref{thm:uQIP-in-NEXP,thm:nexp-in-uqip} gives:

\begin{corollary}
    \label{cor:uqip=nexp}
    $\uQIP[3] = \NEXP$
\end{corollary}

Together, \Cref{thm:uQIP-in-NEXP,thm:nexp-in-uqip} also show that the completeness/soundness gap of any $\uQIP[3]$ protocol can be amplified from inverse-exponential to arbitrarily big:

\begin{corollary}
    \label{cor:uqip_amplification}
    For any $c(n) - s(n) \ge \frac{1}{2^{\poly(n)}}$, $\uQIP[3, c, s] \subseteq \uQIP[3, 1, o(1)]$.
\end{corollary}

\section{Post-Measurement Branching in Unentangled Proof Systems}\label{sec:qmacm}

The preceding section established that $\NEXP \subseteq \uQIP[3]$, in sharp contrast to $\QIP[3] = \PSPACE$. 
This result highlights that the restriction to unentangled proofs can yield proof systems of significantly greater computational power.
On the other hand, it is surprising that the containment of $\NEXP$ in $\uQIP[3]$ made only partial use of the proof system's quantum capabilities, as the second and third messages were purely classical.
To better understand the power of unentangled proof systems, in this section we consider placing further restrictions related to the adaptivity of the verifier.

Concretely, we define a subclass $\uQMACM$ of $\uQIP[3]$ in which the round-2 message consists of public coin tosses and the round-3 message is classical.
A key distinction between $\uQIP[3]$ and $\uQMACM$ is that after the first round the verifier can measure part of the prover’s first message $\ket{\psi}$, obtaining a classical outcome $\ell$ together with the corresponding post-measurement state $\ket{\psi_\ell}$ on the remaining qubits. 
The verifier’s subsequent actions can then depend on $\ell$ and make use of $\ket{\psi_\ell}$. 
We refer to this capability---measuring part a state to extract classical information while retaining the residual post-measurement state---as \emph{post-measurement branching}.
Broadly, our goal in this section is to better understand the power of post-measurement branching in unentangled proof systems.

We formally define $\uQMACM$ below:

\begin{definition}[$\uQMACM$]
A promise problem $A = (A_{\mathrm{yes}}, A_{\mathrm{no}})$ is in $\uQMACM[c, s]$ for polynomial-time computable functions $c, s: \N \to [0,1]$ if there exists a $\QIP[3]$ verification procedure $V$ with the following properties:
\begin{itemize}
  \item \emph{Public coin tosses.} $V_1^x$ samples a random $x \sim \{0,1\}^{\poly(n)}$, records the result in $\mathcal{V}$, and copies it into $\mathcal{M}$.
  \item \emph{Classical final message.} $P_2^x$ sends a classical message.
  Equivalently, $V_2^x$ measures $\mathcal{M}$ in the computational basis before any other processing.
  \item \emph{Completeness.} For all $x\in A_{\mathrm{yes}}$, there exists an \textit{unentangled} prover $P$ that makes $V$ accept with probability at least $c(|x|)$.
  \item \emph{Soundness.} For all $x\in A_{\mathrm{no}}$, every \textit{unentangled} prover $P$ makes $V$ accept with probability at most $s(|x|)$.
\end{itemize}
We define $\uQMACM \coloneqq \uQMACM[2/3, 1/3]$.
\end{definition}

It is not immediately clear whether $\uQMACM$ admits amplification of the completeness and soundness parameters, although we will show its containment in a class that does (\Cref{thm:uQMACM_in_QAM}).

\subsection{Upper Bounding \texorpdfstring{$\uQMACM$}{Unentangled QMACM}}

We now turn to the main result of this section: $\uQMACM \subseteq \QAM$.
For completeness, we also include the definition of $\QAM$, following Marriott and Watrous~\cite{marriott2005quantum}.

\begin{definition}[$\QAM$]\label{def:qam}
A promise problem $A = (A_{\mathrm{yes}}, A_{\mathrm{no}})$ is in $\QAM[c,s]$ if there exists a polynomial-time quantum algorithm $Q(x, y, \ket{\psi})$ such that:
\begin{itemize}
  \item \emph{Completeness.} For all $x\in A_{\mathrm{yes}}$, there exists a collection of $\poly(n)$-qubit states $\{\ket{\psi_y} : y \in \{0,1\}^{\poly(|x|)}\}$ such that $\Pr_y[Q(x, y, \ket{\psi_y}) = 1] \ge c(|x|)$.
  \item \emph{Soundness.} For all $x\in A_{\mathrm{no}}$, for every collection of $\poly(n)$-qubit states $\{\ket{\psi_y} : y \in \{0,1\}^{\poly(|x|)}\}$, $\Pr_y[Q(x, y, \ket{\psi_y}) = 1] \le s(|x|)$.
\end{itemize}
We define $\QAM \coloneqq \QAM[2/3, 1/3]$.
\end{definition}

The definition differs stylistically from our definition of $\QIP[3]$ (\Cref{def:QIP[3]}), which involved the prover and verifier applying alternating quantum channels, only because it is easier to abstract the prover's strategy into a single mapping from strings $y$ to states $\ket{\psi_y}$.
In particular, the string $y$ represents Arthur's random coin tosses sent to Merlin, and $\ket{\psi_y}$ is Merlin's response.
We also remark that $\QAM$ admits amplification by parallel repetition: $\QAM[c, s] = \QAM[1-2^{-\poly}, 2^{-\poly}]$ as long as $c - s$ is inverse-polynomially bounded~\cite[Theorem 4.2]{marriott2005quantum}.

Our proof relies on the existence of an algorithm that compresses a quantum state into a polynomial-sized classical description from which the expectation values of many observables can later be estimated. 
This algorithm derives from a theorem of Aaronson
about simulating bounded-error one-way quantum communication with classical communication, which was in turn used to show that $\BQPqpoly \subseteq \PPpoly$~\cite{aaronson2005limitations}.
In fact, the lemma below is essentially \textit{equivalent} to the non-existence of a superpolynomial quantum advantage in one-way communication complexity for a decision problem.

\begin{lemma}\label{lem:stats-est}
Let $M_1,\ldots,M_{K}$ be $m$-qubit measurement operators (i.e., PSD matrices with $0 \preceq M_i \preceq 1$ for each $i$), and fix an error parameter $\varepsilon$. 
There exist functions $\Stat(\rho, \eps)$ and $\Est(s,i)$ with the following properties:
\begin{enumerate}
  \item $\Stat(\rho, \eps)$ takes as input a classical description of an $m$-qubit state $\rho$ and an error parameter $\eps$, and outputs a classical string of length $O(\log K\frac{m}{\eps^2}\log\frac{m}{\eps})$.
  \item For any $m$-qubit $\rho$ and any $i \in [K]$, $\Est(\Stat(\rho, \eps), i)$ outputs a number that is $\eps$-close to $\tr(\rho M_i)$.
\end{enumerate}
\end{lemma}

\begin{proof}
    Consider a one-way communication problem between two parties, Alice and Bob, in which Alice receives a description of an $m$-qubit state $\rho$, and Bob receives an index $i \in [K]$ and a parameter $t \in [0, 1]$ that is described to $O(\log \frac{1}{\eps})$ bits of precision.
    Their goal is to decide whether $\tr(M_i \rho) \ge t + \eps/10$ or $\tr(M_i \rho) \le t - \eps/10$, promised that one of these is the case.
    
    Observe that this problem admits a bounded-error one-way communication protocol with complexity $O(\frac{m}{\eps^2})$: Alice sends Bob $O(\frac{1}{\eps^2})$ copies of $\rho$, Bob measures $M_i$ on each of the copies, and accepts if and only if the sample mean is greater than $t$.
    By the simulation theorem for quantum one-way communication with classical communication~\cite[Theorem 3.4]{aaronson2005limitations}, this same problem admits a deterministic \textit{classical} communication protocol in which Alice sends $O(\log K\frac{m}{\eps^2}\log\frac{m}{\eps})$ bits to Bob.

    The two functions $\Stat$ and $\Est$ derive directly from this classical communication protocol.
    The encoding function $\Stat$ is simply the function that Alice uses to map $\rho$ to a classical message.
    The decoding function $\Est$ runs Bob's half of the computation for $O(\frac{1}{\eps})$ different values of $t$ to find one that is within $\eps$ of $\tr(\rho M_i)$.
    The correctness of these two functions follows from the correctness of the classical communication protocol.
\end{proof}

We use the compression lemma above to argue that Arthur can ``subsample'' from the second-round messages while approximately maintaining the $\uQMACM$ acceptance probability.
In this next lemma, the $\max_\rho$, $\E_y$, and $\max_z$ operators correspond respectively to the first-, second-, and third-round messages of the $\uQMACM$ protocol.
In plain words, the lemma says that if Arthur samples his message $y$ from a random polynomial-size subset of $\{0,1\}^m$ instead of uniformly over $\{0,1\}^m$, then with high probability over the chosen subset, the completeness probability of the protocol is approximately unchanged.

\begin{lemma}
    \label{lem:subsampling}
    Let $M_{y,z}$ be an $m$-qubit quantum measurement for each $y, z \in \{0,1\}^m$.
    Then except with probability at most $\exp(O(\frac{m^2}{\eps^2}\log\frac{m}{\eps}) - 2\eps^2 r/9))$ over $y_1,\ldots,y_r \sim \{0,1\}^m$,
    \begin{equation}
    \label{eq:subsampling_works}
    \ABS{\max_{\ket{\psi} \in \States(m)} \E_{y \sim \{0,1\}^m} \max_{z \in \{0,1\}^m} \braket{\psi}{M_{y,z}|\psi} - \max_{\ket{\psi} \in \States(m)} \E_{i \sim [r]} \max_{z \in \{0,1\}^m} \braket{\psi}{M_{y_i,z}|\psi}} \le \eps.
    \end{equation}
\end{lemma}

\begin{proof}
    Define
    \[
    f(\ket{\psi}) \coloneqq \E_{y \sim \{0,1\}^m} \max_{z \in \{0,1\}^m} \braket{\psi}{M_{y,z}|\psi}
    \]
    and
    \[
    g(\ket{\psi}) \coloneqq \max_{\ket{\psi} \in \States(m)} \E_{i \sim [r]} \max_{z \in \{0,1\}^m} \braket{\psi}{M_{y_i,z}|\psi}.
    \]
    (The latter is a slight abuse of notation, because $g$ additionally depends on $y_1,\ldots,y_r$.)
    
    Let $\Stat$ and $\Est$ be the functions from \Cref{lem:stats-est}.
    For any string $s$, we claim that except with probability at most $\exp\left(-2\eps^2r / 9 \right)$ over the choices of $y_1, \ldots, y_r$,
    \[
    \E_{y \sim \{0,1\}^m} \max_{z \in \{0,1\}^m} \Est(s, (y, z)) \le \E_{i \sim [r]} \max_{z \in \{0,1\}^m} \Est(s, (y_i, z)) + \frac{\eps}{3}.
    \]
    This is a straightforward consequence of Hoeffding's inequality.
    Thus, for any state $\ket{\psi}$ satisfying $\Stat(\ket{\psi}\!\bra{\psi}, \eps/3) = s$, we have
    \begin{align*}
        f(\ket{\psi})
        &\le \E_{y \sim \{0,1\}^m} \max_{z \in \{0,1\}^m} \Est(s, (y, z)) + \frac{\eps}{3}\\
        &\le \E_{i \sim [r]} \max_{z \in \{0,1\}^m} \Est(s, (y_i, z)) + \frac{2\eps}{3}\\
        &\le g(\ket{\psi}) + \eps,
    \end{align*}
    except with probability at most $\exp\left(-2\eps^2r / 9 \right)$, because of the correctness guarantee of $\Stat$ and $\Est$.
    We similarly obtain $g(\ket{\psi}) \le f(\ket{\psi}) + \eps$ with probability at most $\exp\left(-2\eps^2r / 9 \right)$, which implies $\abs{f(\ket{\psi}) - g(\ket{\psi})} \le \eps$ with probability at most $2\exp\left(-2\eps^2r / 9 \right)$.
    Now, apply a union bound over all possible strings $s$ of length $O(\frac{m^2}{\eps^2}\log \frac{m}{\eps})$ that can be output by $\Stat(\ket{\psi}, \eps/3)$ to conclude that, except with probability at most $\exp(O(\frac{m^2}{\eps^2}\log\frac{m}{\eps}) - 2\eps^2 r/9))$, for \textit{every} $\ket{\psi}$, $\abs{f(\ket{\psi}) - g(\ket{\psi})} \le \eps$.
    From this it follows that $\abs{\max_{\ket{\psi} \in \States(m)}f(\ket{\psi}) - \max_{\ket{\psi} \in \States(m)}g(\ket{\psi})} \le \eps$, which is equivalent to the statement of the lemma.
\end{proof}

Note that in the above lemma, the only use of \Cref{lem:stats-est} is in arguing that we can ``union bound'' over all $m$-qubit quantum states $\ket{\psi}$ as if there were only $2^{\poly(m)}$ possible choices of $\rho$, instead of $2^{2^{\poly(m)}}$.
If $\ket{\psi}$ were an $m$-bit string rather than an $m$-qubit quantum state, then one could \textit{directly} union bound over the choices of $\ket{\psi}$ instead of going through $\Stat(\ket{\psi}\!\bra{\psi}, \eps)$.

We now use \Cref{lem:subsampling} to place $\uQMACM$ in $\QAM$ by combining Merlin's first and third messages into one.
The idea is for Arthur to choose a $\poly(m)$-size subsample of second-round random challenges, to which Merlin responds with both his original first-round quantum message and third-round responses to each of Arthur's random challenges.

\begin{theorem}
    \label{thm:uQMACM_in_QAM}
    For any $c(n) - s(n) \ge \frac{1}{{\poly(n)}}$,
    $\uQMACM[c, s] \subseteq \QAM[c, (c+s)/2] \subseteq \QAM$.
\end{theorem}
\begin{proof}
    Suppose a promise problem $A = (\Ayes, \Ano)$ admits a $\uQMACM[c,s]$ protocol in which given an input $x$ of length $n$, the message register $\mathcal{M}$ has exactly $m = \poly(n)$ qubits.
    As a consequence of \Cref{lem:prover_canonical} and the classical final message, it is without loss of generality that Merlin initially sends Arthur an $m$-qubit state $\ket{\psi}$, Arthur responds with a challenge $y \sim \{0,1\}^m$, and Merlin lastly sends Arthur $z \in \{0,1\}^m$ that is a deterministic function of $y$.
    Let $M_{y,z}$ be the measurement that Arthur applies to $\ket{\psi}$ given $y$ and $z$, so that his acceptance probability is exactly $\braket{\psi|M_{y,z}}{\psi}$.
    Then, Merlin's best strategy achieves acceptance probability exactly
    \begin{equation}
    \label{eq:uQMACM_p_acc}
    \max_{\ket{\psi} \in \States(m)} \E_{y \sim \{0,1\}^m} \max_{z \in \{0,1\}^m} \braket{\psi|M_{y,z}}{\psi}.
    \end{equation}

    Now consider the following $\QAM$ protocol for $A$.
    For some $r = \poly(n)$ to be chosen later, Arthur starts by sending Merlin $m \cdot r$ random bits, which we interpret as challenge strings $y_1,\ldots,y_r$ each of length $m$.
    Merlin's response consists of an $m$-qubit state $\ket{\psi}$ and strings $z_1,\ldots,z_r \in \{0,1\}^m$.
    Arthur picks $i \in [r]$ at random, then measures $M_{y_i, z_i}$ on $\rho$.
    The maximum acceptance probability of this $\QAM$ protocol is
    \[
    \E_{y_1,\ldots,y_r \sim \{0,1\}^m} \max_{\substack{\ket{\psi} \in \States(m)\\z_1,\ldots,z_r \in \{0,1\}^m}} \E_{i \in [r]} \braket{\psi|M_{y_i,z_i}}{\psi}.
    \]
    This is evidently equal to
    \begin{equation}
    \label{eq:QAM_p_acc}
    \E_{y_1,\ldots,y_r \sim \{0,1\}^m} \max_{\ket{\psi} \in \States(m)} \E_{i \in [r]} \max_{z \in \{0,1\}^m} \braket{\psi|M_{y_i,z}}{\psi}
    \end{equation}
    because for fixed $\rho$, the choice of $z_i$ only affects the $i$th term in the inner expectation, so to maximize the expectation it is optimal to maximize each term separately.

    We claim that the $\QAM$ protocol has completeness at least $c$.
    Given Merlin's optimal strategy for the $\uQMACM$ protocol,
    Merlin's strategy for the $\QAM$ protocol is to choose the same $\ket{\psi}$ that he would in the $\uQMACM$ protocol, and then for each $i \in [r]$ to let $z_i$ be the string $z$ that he would send upon receiving $y = y_i$ from Arthur.
    This strategy for the $\QAM$ protocol causes Arthur to accept with probability $c$ given any fixed $i \in [r]$, so the acceptance probability is certainly still $c$ when averaging over $i \in [r]$.

    It remains to bound the soundness.
    We do so by upper bounding the acceptance probability of the $\QAM$ protocol (\Cref{eq:QAM_p_acc}) in terms of that of the $\uQMACM$ protocol (\Cref{eq:uQMACM_p_acc}), via application of \Cref{lem:subsampling}.
    Let $p \le \exp(O(\frac{m^2}{\eps^2}\log\frac{m}{\eps}) - 2\eps^2 r/9))$ be the probability that \Cref{eq:subsampling_works} fails to hold in \Cref{lem:subsampling}.
    Then the $\QAM$ acceptance probability is bounded by
    \[
        \E_{y_1,\ldots,y_r \sim \{0,1\}^m} \max_{\rho \in \States(m)} \E_{i \in [r]} \max_{z \in \{0,1\}^m} \tr(\rho M_{y_i,z_i})
        \le p + \eps + \max_{\rho \in \States(m)} \E_{y \sim \{0,1\}^m} \max_{z \in \{0,1\}^m} \tr(\rho M_{y,z}).
    \]
    Choose $\eps = (c - s)/4$ and $r = O(\frac{m^2}{\eps^4} \log \frac{m}{\eps}) \le \poly(|x|)$ so that $p \le \eps$.
    Then assuming $x \in \Ano$, the right hand side is at most $(c - s)/4 + (c - s)/4 + s = (c + s)/2$.

    To summarize, we have shown that $A$ admits a $\QAM[c, (c+s)/2]$ protocol.
    The completeness/soundness can be amplified to $[2/3, 1/3]$, or even further to $[1 - 2^{-\poly}, 2^{-\poly}]$ as shown by Marriott and Watrous~\cite[Theorem 4.2]{marriott2005quantum}.
\end{proof}

Observe that nowhere in the proof was it crucial for Arthur's message to consist of uniformly random public coins.
One could instead envision a class ``$\mathsf{qcc}$-$\uQIP^{\mathsf{no\text{-}pmb}}[3]$'' between $\uQMACM$ and $\uQIP[3]$ in which the verifier's initial classical message is produced by an efficient quantum procedure acting only on the register $\mathcal{V}$, but the final prover message remains classical.
Formally, that would mean $V_1^x$ performs a measurement on $\mathcal{V}$ independent of the received message on $\mathcal{M}$, then swaps the measurement with the prover's message in $\mathcal{M}$.
Intuitively, $\mathsf{qcc}$-$\uQIP^{\mathsf{no\text{-}pmb}}[3]$ is like $\QIP[3]$ except that only the first message is quantum and the verifier is not allowed to use post-measurement branching on the first quantum message.
Then the same proof above would show containment of $\mathsf{qcc}$-$\uQIP^{\mathsf{no\text{-}pmb}}[3]$ in $\QIP[2]$.
We did not attempt to define such a complexity class formally because it seems impossible to give it a short but sufficiently descriptive name.

\section{Constant-Round Quantum-Classical Interactive Proofs}\label{sec:qcip-qcam}

\cref{sec:uqip3} and \cref{sec:qmacm} illustrate that post-measurement branching is an interesting feature of quantum proof systems, and in particular one of the key ingredients that enables three-round unentangled protocols to capture $\NEXP$.
In this section, we turn to the simplest setting where post-measurement branching can be isolated: the classes $\QCAM$ and $\QCIP[2]$. 

Both $\QCAM$ and $\QCIP[2]$ are two-round interactions initiated by Arthur.
In $\QCAM$, Arthur’s message consists of uniformly random classical bits, whereas in $\QCIP[2]$ Arthur may prepare a state $\ket{\psi}$, measure part of it to obtain a message $\ell$ and the post-measurement state $\ket{\psi_\ell}$, and then use both in the remainder of the verification procedure.

The main results of this section are summarized in the following theorem.

\begin{theorem}\label{thm:qcam-qcip-difference}
   $\QCAM = \BP \cdot \QCMA \subseteq \BQ \cdot \QCMA \subseteq \QCIP[2] \subseteq \BQP^{\NP^{\PP}}$. 
\end{theorem}

We will record formal definitions shortly.

\subsection{Quantum-Classical Arthur-Merlin Games}
We begin by establishing that $\QCAM = \BP \cdot \QCMA$. This equivalence is not new: as far as we can tell, it first appeared in Marriott’s Master’s thesis~\cite[Theorem 8]{marriott2003nondeterminism}. The upper bound $\QCAM \subseteq \BP \cdot \QCMA$ was later reproved in \cite[Proposition 33]{lockhart2017quantumstateisomorphism}, and the equivalence was stated without proof or citation in \cite{aaronson2023certified}. 
We include the argument here to present a complete proof in more standard notation than that of Marriott’s thesis.

We start by defining $\QCAM$ and the $\BP$ operator. 

\begin{definition}[$\QCAM$]
A promise problem $A = (A_{\mathrm{yes}}, A_{\mathrm{no}})$ is in $\QCAM[c,s]$ if there exists a polynomial-time quantum algorithm $Q(x, y, z)$ such that:
\begin{itemize}
  \item \emph{Completeness.} For all $x\in A_{\mathrm{yes}}$, there exists a collection of $\poly(n)$-length bit strings $\{z_y : y \in \{0,1\}^{\poly(|x|)}\}$ such that $\Pr_y[Q(x, y, z_y) = 1] \ge c(|x|)$.
  \item \emph{Soundness.} For all $x\in A_{\mathrm{no}}$, for every collection of $\poly(n)$-length bit strings $\{z_y : y \in \{0,1\}^{\poly(|x|)}$, $\Pr_y[Q(x, y, z_y) = 1] \le s(|x|)$.
\end{itemize}
We define $\QCAM \coloneqq \QCAM[2/3, 1/3]$.
\end{definition}

Similar to $\QAM$ (\cref{def:qam}), the string $y$ represents Arthur's random coin tosses sent to Merlin, and $z_y$ is Merlin's response.
One can more generally define $\QCAM[k]$ with $k \geq 2$ rounds of interaction. 
However, this class collapses to $\QCAM$ \cite[Theorem 7(iv)]{doi:10.1137/17M1160173}, in analogy with the collapse of the $\AM$ hierarchy \cite{babaicollapse,BM88-am}.
$\QCAM$ also admits amplification by running many iid trials in parallel and taking the majority vote: $\QCAM[c,s] = \QCAM[1-2^{-\poly},2^{-\poly}]$ as long as $c-s$ is inverse-polynomially bounded.
A proof of this amplification trick is straightforward, and is also a special case of~\cite[Lemma 20]{doi:10.1137/17M1160173}.

We now define the $\BP$ operator, which has a few essentially equivalent definitions in the literature (see~\cite[Section 3.1]{buhrman2024classicalversusquantumqueries}).
For our purposes, the following definition is most convenient.

\begin{definition}[$\BP$ operator]
\label{def:bp_operator}
Let $\mathcal{C}$ be any class of promise problems.
Then $\BP \cdot \mathcal{C}$ consists of all promise problems $A = (A_{\mathrm{yes}}, A_{\mathrm{no}})$ for which there exist a promise problem $B = (B_{\mathrm{yes}}, B_{\mathrm{no}}) \in \mathcal{C}$ and a polynomial $p$ such that, for every input $x$ of length $n$,
\begin{itemize}
    \item Completeness: If $x \in \Ayes$, then 
    $\Pr_{y}[(x,y) \in \Byes] \geq 2/3,$
    \item Soundness: If $x \in \Ano$, then 
    $\Pr_{y}[(x,y) \in \Bno] \geq 2/3,$
\end{itemize}
where $y \in \{0,1\}^{p(n)}$ is uniformly distributed.
\end{definition}

\begin{theorem}\label{thm:qcam=bp}
$\QCAM = \BP \cdot \QCMA$.
\end{theorem}
\begin{proof}
    Let $A$ be a promise problem in $\BP \cdot \QCMA$ and let $B$ be the corresponding promise problem in $\QCMA$ to which $A$ reduces under \Cref{def:bp_operator}.
    Consider a $\QCMA$ verifier $Q(x, y, z)$ for $B$ where $(x, y)$ is the input to $B$ and $z$ is the witness.
    Suppose without loss of generality that completeness and soundness parameters of $Q$ are amplified to $0.9$ and $0.1$, respectively.
    We claim that $Q$ is also a $\QCAM$ verifier for $A$ with completeness $0.6$ and soundness $0.4$.
    In particular:
    \begin{enumerate}
        \item If $x \in \Ayes$, then $\Pr_y[(x, y) \in \Byes] \ge 2/3$, and therefore $\E_y[\max_z \Pr[Q(x, y, z) = 1)]] \ge 2/3 \cdot 0.9 = 0.6$.
        \item If $x \in \Ano$, then $\Pr_y[(x, y) \in \Bno] \ge 2/3$, and therefore $\E_y[\max_z \Pr[Q(x, y, z) = 1)]] \le 1/3 \cdot 1 + 2/3 \cdot 0.1 = 0.4$.
    \end{enumerate}
    Put another way, the $\QCAM$ protocol is for Arthur to send Merlin random coin tosses $y$ for which $(x, y)$ forms an instance of $B$, and then to run the $\QCMA$ verifier on $(x, y)$ and Merlin's response $z$.
    Hence $A \in \QCAM[0.6, 0.4] \subseteq \QCAM$ because one can amplify $(0.6, 0.4)$ to $(2/3, 1/3)$ by parallel repetition.
    This establishes $\BP \cdot \QCMA \subseteq \QCAM$.

    For the converse, let $A$ be a promise problem in $\QCAM$.
    Take a $\QCAM$ verifier $Q(x, y, z)$ with completeness $0.9$ and soundness $0.1$.
    Let $B$ be the $\QCMA$ promise problem parameterized by $Q$, where $(x, y)$ is interpreted as the input to $B$ and $z$ is the witness.
    That is, $(x, y) \in \Byes$ if there is a $z$ that causes $Q(x, y, z)$ to accept with probability at least $2/3$, and $(x, y) \in \Bno$ if every $z$ causes $Q(x, y, z)$ to accept with probability at most $1/3$.
    We claim that $B$ shows $A \in \BP \cdot \QCMA$.
    In particular:
    \begin{itemize}
        \item If $x \in \Ayes$, then $\E_y[\max_z \Pr[Q( x, y, z) = 1]] \ge 0.9$.
        Because $\Pr[Q(x, y, z)] \in [0, 1]$, it must be the case that $\Pr_y[\max_z \Pr[Q(x, y, z) = 1] \ge 2/3] \ge 0.7$, and therefore $\Pr_y[(x, y) \in \Byes] \ge 0.7$.
        \item If $x \in \Ano$, then $\E_y[\max_z \Pr[Q, x, y, z) = 1]] \le 0.1$.
        Because $\Pr[Q(x, y, z)] \in [0, 1]$, it must be the case that $\Pr_y[\max_z \Pr[Q(x, y, z) = 1] \le 1/3] \ge 0.7$, and therefore $\Pr_y[(x, y) \in \Bno] \ge 0.7$.
    \end{itemize}
    This shows that $B$ satisfies \Cref{def:bp_operator} with respect to $A$.
\end{proof}

We note that a similar result, $\QAM = \BP \cdot \QMA$, was also shown in Marriott's thesis \cite[Theorem 12]{marriott2003nondeterminism}.
A proof follows from an identical argument to that of \cref{thm:qcam=bp}.

\subsection{Quantum-Classical Interactive Proofs}
Here we establish $\BQ \cdot \QCMA \subseteq \QCIP[2] \subseteq \BQP^{\NP^{\PP}}$.
We start with some definitions.

A $\QCIP[2]$ verification procedure is specified by a polynomial-time uniformly generated family of quantum circuits $V = \{V_1^x, V_2^x : x \in \{0,1\}^*$.
On input $x$ of length $n$, these circuits determine the verifier's actions across the two-round interaction.
Each circuit acts on registers of size $\poly(n)$, partitioned into a message register $\calM$, exchanged with the prover, and a verifier workspace register $\calV$, retained by the verifier. 
The prover is modeled by an unrestricted family of quantum operations $P = \{P^x : x \in \{0,1\}^*\}$ that act on the same message register $\calM$ and a \emph{prover workspace} register $\calP$ (which need not be polynomially-bounded in size). 

The defining restriction of $\QCIP[2]$ is that \emph{the message register is measured in the computational basis before transmision}. Equivalently, the final operation of both $V_1^x$ and $P^x$ on $\calM$ is a computational basis measurement. 
The interaction proceeds as follows: 
\begin{enumerate}
\item The three registers $\calV, \calM, \calP$ are initialized to the all-zeros state.
\item The verifier applies $V_1^x$ to $\calP$ and $\calM$.
The register $\mathcal{M}$ is then measured in the computational basis and sent as a classical message.
\item The prover applies $P^x$ to $\calP$ and $\calM$. Again, $\mathcal{M}$ is measured in the computational basis before being sent back.
\item Finally, the verifier applies $V_2^x$ to $(\calV$, $\calM$) and measures a designated output qubit to decide acceptance or rejection.
\end{enumerate}

We now formally define the class $\QCIP[2]$.

\begin{definition}[{$\QCIP[2]$}]
   A promise problem $A = (A_{\rm yes}, A_{\rm no})$ is in $\QCIP[2, c, s]$ for polynomial-time computable functions $c,s : \N \to [0,1]$ if there exists a $\QCIP[2]$ verification procedure $V$ with the following properties: 
   \begin{itemize}
       \item Completeness: For all $x \in A_{\rm yes}$, there exists a quantum prover $P$ that causes $V$ to accept $x$ with probability at least $c(\abs{x})$.
       \item Soundness: For all $x \in A_{\rm no}$, every quantum prover $P$ causes $V$ to accept $x$ with probability at most $s(\abs{x})$. 
   \end{itemize}
We define $\QCIP[2] \coloneqq \QCIP[2,2/3,1/3]$.
\end{definition}

We note that error reduction in $\QCIP(2)$ can be achieved by parallel repetition: running multiple independent instances of the protocol in parallel and deciding acceptance by a majority vote.

Next, we define the $\BQ$ operator, which was recently formalized by Buhrman, Le Gall, and Weggemans~\cite{buhrman2024classicalversusquantumqueries}.Intuitively, $\BQ \cdot \calC$ consists of those promise problems that admit polynomial-time \emph{quantum} reductions to problems in $\calC$.

\begin{definition}[$\BQ$ operator]
\label{def:bq_operator}
Let $\mathcal{C}$ be any class of promise problems.
Then $\BQ \cdot \mathcal{C}$ consists of all promise problems $A = (A_{\mathrm{yes}}, A_{\mathrm{no}})$ for which there exist a promise problem $B = (B_{\mathrm{yes}}, B_{\mathrm{no}}) \in \mathcal{C}$ and a polynomial-time quantum algorithm $\calA$ such that, for every input $x$ of length $n$,
\begin{itemize}
    \item Completeness: If $x \in \Ayes$, then 
    $\Pr_{y\sim \calA(x)}[(x,y) \in \Byes] \geq 2/3,$
    \item Soundness: If $x \in \Ano$, then 
    $\Pr_{y\sim \calA(x)}[(x,y) \in \Bno] \geq 2/3.$
\end{itemize}
Here the distribution $y \sim \calA(x)$ is obtained by running $\calA$ on input $x$ and measuring a designated polynomial-size output register in the computational basis. 
\end{definition}

We will now prove lower and upper bounds on $\QCIP[2]$. 

\begin{theorem}\label{thm:qcip2-bounds}
    $\BQ \cdot \QCMA \subseteq \QCIP[2].$
\end{theorem}
\begin{proof}
Let $A$ be a promise problem in $\BQ \cdot \QCMA$. Let $B$ be the corresponding promise problem in $\QCMA$ to which $A$ reduces via the polynomial-time quantum algorithm $\calA$ under \Cref{def:bp_operator}.
Consider a $\QCMA$ verifier $Q(x, y, z)$ for $B$ where $(x, y)$ is the input to $B$ and $z$ is the witness.
Suppose without loss of generality that completeness and soundness parameters of $Q$ are amplified to $0.9$ and $0.1$, respectively.
We claim that $Q$ is also a $\QCIP[2]$ verifier for $A$ with completeness $0.6$ and soundness $0.4$.
In particular:
    \begin{enumerate}
        \item If $x \in \Ayes$, then $\Pr_{y\sim \calA(x)}[(x, y) \in \Byes] \ge 2/3$, and therefore $\E_{y\sim \calA(x)}[\max_z \Pr[Q(x, y, z) = 1)]] \ge 2/3 \cdot 0.9 = 0.6$.
        \item If $x \in \Ano$, then $\Pr_{y\sim \calA(x)}[(x, y) \in \Bno] \ge 2/3$, and therefore $\E_{y \sim \calA(x)}[\max_z \Pr[Q(x, y, z) = 1)]] \le 1/3 \cdot 1 + 2/3 \cdot 0.1 = 0.4$.
    \end{enumerate}
The protocol is for Arthur to send to Merlin $y = \calA(x)$ for which $(x, y)$ forms an instance of $B$, and then to run the $\QCMA$ verifier on $(x, y)$ and Merlin's response $z$.
Hence $A \in \QCIP[2,0.6, 0.4] \subseteq \QCIP[2]$ because one can amplify $(0.6, 0.4)$ to $(2/3, 1/3)$ by parallel repetition.
This establishes $\BQ \cdot \QCMA \subseteq \QCIP[2]$.
\end{proof}

\begin{theorem}
\label{thm:QCIP[2]_upper}
    $\QCIP[2] \subseteq \BQP^{\NP^{\PP}}$.
\end{theorem}

\begin{proof}
Let $V = \{V_1^x, V_2^x\}$ be a verifier with completeness $0.9$ and soundness $0.1$ for some promise problem $A \in \QCIP[2]$.
The $\PP$ language will be specified by a $\PostBQP$ promise problem, because $\PP = \PostBQP$~\cite{aaronson2005quantum} and any $\PP$ promise problem can be extended to a $\PP$ language, as $\PP$ is a syntactic class.
Given a tuple $(x, y, z)$, consider a $\QCIP[2]$ interaction between verifier and prover in which we postselect on $V_1^x$ sending the classical message $y$, and the prover responds with $z$.
Then deciding whether this postselected interaction between prover and verifier causes $V$ to accept with probability at least $2/3$ (yes) or at most $1/3$ (no) is clearly in $\PostBQP = \PP$.
The $\NP^\PP$ language, then, will be: given $(x, y)$, decide whether there exists a string $z$ for which $(x, y, z)$ is a yes instance of the $\PP$ language.
Finally, consider the following $\BQP^{\NP^\PP}$ machine that we claim decides $A$: run $V_1^x$ to obtain $y \in \{0,1\}^{\poly(|x|)}$, query the $\NP^\PP$ language on $(x, y)$, and accept if and only if $(x, y)$ is a yes-instance.
This machine works because:
\begin{itemize}
    \item If $x \in \Ayes$, then $\E_{y \sim V_1^x}[\max_z \Pr[V \text{ accepts } z \mid y]] \ge 0.9$, and therefore $\Pr_{y \sim V_1^x}[\max_z \Pr[V \text{ accepts } z \mid y] \ge 2/3] \ge 0.7$.
    Hence, with probability at least $0.7$ over the $y$ sampled by the $\BQP$ machine, there exists a $z$ for which $(x, y, z)$ is a yes-instance of the $\PP$ language, and thus the $\BQP$ machine accepts.
    \item If $x \in \Ano$, then $\E_{y \sim V_1^x}[\max_z \Pr[V \text{ accepts } z \mid y]] \le 0.1$, and therefore $\Pr_{y \sim V_1^x}[\max_z \Pr[V \text{ accepts } z \mid y] \le 1/3] \ge 0.7$.
    Hence, with probability at least $0.7$ over the $y$ sampled by the $\BQP$ machine, for every $z$ $(x, y, z)$ is a no-instance of the $\PP$ language, and thus the $\BQP$ machine rejects.
\end{itemize}
So, the $\BQP^{\NP^\PP}$ machine decides $A$ with error probability at most $0.4$, which can of course be amplified to arbitrarily small error.
\end{proof}

We conclude with several remarks regarding \cref{thm:qcip2-bounds}.
First, unlike $\BP \cdot \QCMA = \QCAM$, the class $\QCIP[2]$ is not characterized by $\BQ \cdot \QCMA$. The reason is that $\BQ \cdot \QCMA$ does not capture the additional power conferred by post-measurement branching.  
In fact, one can show that the subclass of $\QCIP[2]$ in which post-measurement branching is disallowed is precisely equal to $\BQ \cdot \QCMA$. 

Second, our upper bound generalizes straightforwardly to $k$-round interactions.
Specifically, one can show that $\QCIP[2k]$ is contained in a tower of classes of the form 
\begin{equation*}
\BQP^{\NP^{\PP^{\NP^{\PP^{\dots}}}}},
\end{equation*}
which in particular implies that $\QCIP[k] \subseteq \CH$ for any constant $k$.

Finally, by the same reasoning as in \cref{thm:qcip2-bounds}, one can show that $\BQ \cdot \QMA \subseteq \QIP[2]$.  
However, obtaining an upper bound on $\QIP[2]$ stronger than $\PSPACE$ is unclear: our proof technique in \cref{thm:QCIP[2]_upper} breaks when Merlin's message is quantum, as it is in $\QIP[2]$.

\section*{Acknowledgments}
We thank Scott Aaronson for clarifications regarding \cref{lem:stats-est}.

S.G. is supported in part by an IBM Ph.D.\ Fellowship. 
W.K. is supported by the U.S. Department of Energy, Office of Science, National Quantum Information Science Research Centers, Quantum Systems Accelerator, and by NSF Grant CCF-231173.
This work was done in part while S.G. was visiting the Simons Institute for the Theory of Computing, supported by NSF Grant QLCI-2016245. 

\printbibliography

\end{document}